\documentclass[onecolumn,draftcls]{IEEEtran}
\ifCLASSINFOpdf
\else
\fi
\usepackage[cmex10]{amsmath}
\usepackage{algorithm}
\usepackage[noend]{algcompatible}

\usepackage{array}
\usepackage{amsmath,amssymb,amsthm,array,graphicx,verbatim}
\usepackage{color,array,graphicx,verbatim,xtab}
\usepackage{setspace}
\usepackage[breaklinks,bookmarks,bookmarksnumbered,bookmarksopen,bookmarksopenlevel=2]{hyperref}

\newtheorem{theorem}{Theorem}
\newtheorem{lemma}{Lemma}

\newtheorem{remark}{Remark}

\usepackage{enumerate}
\usepackage[numbers, sort&compress]{natbib}
\usepackage{epstopdf}

\usepackage{tikz}
\usetikzlibrary{shapes}
\tikzstyle{block} = [draw,rectangle, rounded corners, minimum width=1cm, minimum height=0.8cm,text centered, line width=2pt ]
\tikzstyle{arrow} = [thick,->,>=stealth,line width=2pt]

\newcommand{\g}{\text{CIMA}}

\begin{document}
%
\title{A Common Information-Based Multiple Access Protocol Achieving Full Throughput and Linear Delay}

%

\author{Yi~Ouyang,~\IEEEmembership{Student Member,~IEEE} and Demosthenis~Teneketzis,~\IEEEmembership{Fellow,~IEEE}
\thanks{This paper was presented in part at \textit{the 2015 IEEE International Symposium on Information Theory} \cite{ouyangISIT}. 
}
\thanks{Y. Ouyang and D. Teneketzis are with the Department of Electrical Engineering and Computer Science, University of Michigan, Ann Arbor, MI (e-mail: ouyangyi@umich.edu; teneket@umich.edu).}
 }

\maketitle

\begin{abstract}
We consider a multiple access communication system where multiple users share a common collision channel.
Each user observes its local traffic and the feedback from the channel.
At each time instant the feedback from the channel is one of three messages: no transmission, successful transmission, collision.
The objective is to design a transmission protocol that coordinates the users' transmissions and achieves high throughput and low delay.

We present a decentralized Common Information-Based Multiple Access (CIMA) protocol that has the following features:
(i) it achieves the full throughput region of the collision channel;
(ii) it results in a delay that is linear in the number of users, and is significantly lower than that of CSMA protocols;
(iii) it avoids collisions without channel sensing.
\end{abstract}


%
\IEEEpeerreviewmaketitle

\section{Introduction}

Multiple access communication has played a crucial role in the operation of many networked systems, including satellite networks, radio networks, wired/wireless Local Area Networks (LANs), and data centers.
One important feature of multiple access communication is its decentralized information structure.
In general, when multiple users share the communication system,
coordination among them is essential to resolve collision issues.
In the absence of a centralized controller, it is challenging to design efficient user coordination mechanisms.

We consider a typical slotted multiple access communication system where multiple users share a common collision channel.
Each user is equipped with an infinite size buffer and observes Bernoulli arrivals to its own queue.
In addition to the local information, all users receive a common broadcast feedback from the channel.
The feedback indicates whether the previous transmission was successful (exactly one user transmitted), or it was a collision (more than one users transmitted), or the channel was idle. The objective is to design a transmission protocol that effectively coordinates the users' transmissions under the above described information structure.
In the design of transmission protocols, there are two major performance metrics of interest: throughput and delay.
The throughput region of a protocol is the set of arrival rates for which the users' queues are stable (see detailed definition in Section \ref{sub:model:throughput}) under the protocol. The delay performance of a protocol is the average waiting time of a packet in the communication system. An efficient transmission protocol should achieve the maximum throughput region and incur low transmission delay.

In this paper, we propose a common information (see \cite{nayyar2013decentralized,nayyar2014common}) based multiple access protocol ($\g$) that uses the common channel feedback to coordinate users. 
In $\g$, each user constructs upper bounds on the lengths of the queues of all users, including itself, based on previous transmission strategies and the common feedback.
Since the upper bounds are common knowledge, users can coordinate their transmission through these common upper bounds to avoid collision. 
We prove that without knowledge of any statistics, $\g$ achieves the full throughput region of the collision channel.
We also prove that the $\g$ protocol incurs low transmission delay; the delay is upper-bounded by a linear function of the number of users.

There is a rich literature on multiple access communications. 
Below we present a survey of this literature.

\subsection*{Related Work}

There are primarily two classes of protocols for the situation where the alphabet of the feedback channel is 
$\{0,1,e\}= \{ $no transmission, successful transmission, collision$\}$:
collision-free and contention-based protocols.
Time-division-multiple-access (TDMA) \cite{rom1990multiple} and adaptive TDMA \cite{papadimitriou2003adaptive, papadimitriou2006high} are collision-free protocols.
In adaptive TDMA protocols the (common) information provided by the feedback is used to adaptively coordinate users to avoid collision.
Adaptation resolves the problems due to asymmetric arrivals, and collision avoidance results in higher throughput and lower delay than TDMA. However, there is no theoretical analysis of adaptive TDMA protocols.
Backoff-type protocols and Aloha protocols \cite{rom1990multiple} allow for contention/collision.
Due to collision, most contention based protocols can not achieve full throughput.
However, polynomial back-off protocols, presented and analyzed in \cite{haastad1996analysis}, achieve full throughput. Nevertheless, polynomial back-off protocols have exponential delay performance in simulation.

Several types of multiple access protocols were proposed when the common information among the users is more than $\{0,1,e\}$.
The authors of \cite{shah2011medium, rajagopalan2009network, shah2012randomized} proposed decentralized random access protocols that achieve full throughput when each user knows the maximum queue length in the system or all other users' transmission results.
When channel sensing is allowed, carrier sense multiple access (CSMA) protocols, proposed in \cite{jiang2010distributed, jiang2011approaching,ni2012q, ghaderi2013fundamental,jiang2012fast,lee2014provable}, achieve full throughput when the channel sensing portion of time is not taken into account in the throughput calculation.
A survey of CSMA protocols is presented in \cite{yun2012optimal}.
In terms of delay performance, the CSMA protocols proposed in \cite{jiang2012fast,lee2014provable} achieve delay that is linear in the number of users. 

Multiple access protocols for adversarial queueing models were presented in \cite{chlebus2012adversarial,anantharamu2009adversarial}.
In \cite{chlebus2012adversarial,anantharamu2009adversarial} it is proved that these protocols achieve full throughput and have linear delay in the number of users.

Other models for multiple access have also been proposed in the literature.
In \cite{wang2014optimal}, channel switching policies that achieve high throughput for multiple access have been considered within the context of the slotted Aloha protocol and the IEEE 802.11 WLANs protocol.
The stability region of the multi-packet reception multiple access channel has been investigated in \cite{luo2006throughput}.
Multiple access with noisy channels has been considered in \cite{ying2011throughput, reddy2012distributed}, and the stability region of policies with delayed shared information has been determined.

\subsection*{Contributions of the Paper}
We present a collision-free protocol ($\g$) that achieves full throughput and delay that is linear in the number of users.
The protocol is based on the common information approach to decentralized decision-making \cite{nayyar2013decentralized}.
The common information in our problem is the feedback, $0,1$ or $e$, provided at each time instant to all the users by the collision channel. The protocol achieves lower delay than adaptive TDMA and back-off protocols. It also achieves lower delay than CSMA. CSMA protocols achieve delay that is linear in the number users, but it is significantly higher than that of $\g$. Furthermore, CSMA protocols require more communication and coordination among users than the $\g$ protocol.
The $\g$ protocol is simple to implement, as at each time instant it only requires knowledge of the upper bounds on each user's queue length. The upper bounds on the users' queue lengths are common knowledge and are updated in a simple manner.

\subsection*{Organization}
The rest of the paper is organized as follows. In Section \ref{sec:model} we present the system model and formulate the problem under investigation. In section \ref{sec:CIMA} we present the $\g$ protocol.
In Section \ref{sec:analysis} we prove that the $\g$ protocol achieves full throughput and linear (in the number of users) delay.
We present simulation results and compare the delay of our protocol with the delay of other 
protocols that achieve full throughput in Section \ref{sec:simulation}.
We conclude in Section \ref{sec:conclusion}.
We present the proof of the technical results in Appendices A-E.
\subsection*{Notation}
Random variables are denoted by upper case letters, their realization by the corresponding lower case letter.
In general, subscripts are used as time index while superscripts are used to index users. 
For time indices $t_1\leq t_2$, $X_{t_1:t_2}$ is the short hand notation for $(X_{t_1},X_{t_1+1},...,X_{t_2})$.
For a policy/protocol $g$, we use $X^{g}$ to indicate that the random variable $X^{g}$ depends on the choice of policy $g$.
$\mathbf{P}(\cdot)$ is the probability of an event. For random variables $X,Y$ with realizations $x,y$, $\mathbf{P}(x|y) := \mathbf{P}(X=x|Y=y)$.
For a policy $g$ and a parameter $\lambda$, $\mathbf{P}^{\lambda,g}(\cdot)$ indicates that the probability depends on the policy $g$ and the parameter $\lambda$.

\section{System Model and Objective}
\label{sec:model}
\subsection{System Model}
\label{sub:model}
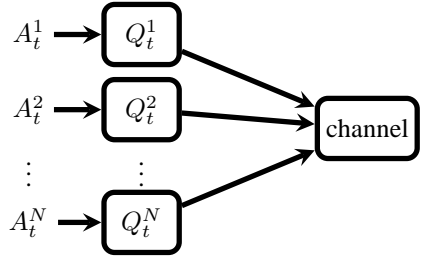
\begin{figure}
\begin{center}
\begin{tikzpicture}[node distance=1.5cm]

\node (Q1) at (1,2.5) [block] {$Q^1_t$};
\node (Q2) at (1,1.5) [block] {$Q^2_t$};
\node (Q3) at (1,0.75) {$\vdots$};
\node (QN) at (1,0) [block] {$Q^N_t$};
\node (CH) at (4,1.25)[block] {channel};
\node (I1) [left of =Q1] {$A^1_t$};
\node (I2) [left of =Q2] {$A^2_t$};
\node (I3) [left of = Q3] {$\vdots$};
\node (IN) [left of =QN] {$A^N_t$};

\draw [arrow]  (I1) -- (Q1);
\draw [arrow]  (I2) -- (Q2);
\draw [arrow]  (IN) -- (QN);

\draw [arrow]  (Q1) -- (CH);
\draw [arrow]  (Q2) -- (CH);
\draw [arrow]  (QN) -- (CH);
\end{tikzpicture}

\end{center}

\caption{Multiple Access Collision Channel}
\label{fig:system}
\end{figure}
We consider a slotted communication system, described by Fig. \ref{fig:system}, 
where $N$ users, indexed by $1,2,\dots,N$, 
share a common collision channel; we denote by $\mathcal{N}:=\{1,2,\dots,N\}$ the set of channel users.
Each user $n$ is associated with an infinite size buffer
with queue length $Q^n_t$ at the beginning of each time slot $t$.
We assume that each queue is initially empty.

At each time slot $t$ each user can transmit one packet in its queue through the shared channel.
If only one user transmits in a time slot, the transmission is successful and the transmitted packet is removed form the queue; if more than one users transmit simultaneously, a collision occurs and all packets involved in the collision remain in their queue.
We consider Bernoulli arrivals to the system.
Let $A^n_t$ denote the packet arrival to user $n$ at time $t$;
$A^n_t=1$ means that a packet arrives at queue $n$ right after the transmission at time $t$.
The arrival $A^n_t$ is a Bernoulli random variable with parameter $\lambda^n$, and the
arrival processes
$\{A^n_t, t=0,1,\dots\}, n \in \mathcal{N}$ are independent.
Let $U^n_t$ denote the transmission decision of user $n$ at time slot $t$; $U^n_t=1$ (resp. $0$) indicates that user $n$ transmits (resp. does not transmit) at time $t$. 
The dynamics of queues are given by
\begin{align}
Q^n_{t+1} = A^n_t+\left(Q^n_t - U^n_t\prod_{m\neq n}\left(1-U^m_t\right) \right)^+,
\label{eq:qdynamics}
\end{align}
where $(\cdot)^+ := \max(\cdot,0)$.
We assume that at the end of each time slot $t$, every user receives a feedback $F_t\in \{0,1,e\}$ from the channel/receiver indicating whether no packets, one packet, or more than one packet (a collision) were transmitted, respectively, in this time slot.
This communication system is decentralized; each user can only observe its own queue length, its arrivals and the common feedback.
Moreover, the arrival rates $\lambda := (\lambda^1,\lambda^2,\dots,\lambda^N)$ are \textit{not known} to the users.
Therefore, the users' decisions according to any decentralized transmission policy/protocol $g=\{g^n_t, n=1,2,\dots,N, t=0,1,\dots\}$ are generated by
\begin{align}
U^n_t = g^n_t(Q^n_{0:t}, A^n_{0:t-1}, U^n_{0:t-1}, F_{0:t-1}),
\label{eq:policy}
\end{align}
$ n=1,2,\dots,N, t=0,1,2,\dots$

In this paper, we consider throughput and queueing delay as the performance metrics of a decentralized transmission policy/protocol.
The objective is to design a decentralized protocol to achieve full throughput and to maintain low queueing delay. We proceed to define the throughput region and queueing delay of the communication system.

\subsection{Stability and Throughput Optimality}\label{sub:model:throughput}

For queueing systems that can be described by irreducible Markov chains,
stability is usually defined to be positive recurrence of the corresponding Markov chains.
In this problem, the users' actions can generally depend on the whole history of information.
When non-Markovian control policies are used, the resulting queue length processes are not Markov in general.
Even within the class of Markovian policies, the corresponding Markov chain may not be irreducible under any Markovian policy.

To achieve higher throughput performance of the communication system, we consider general 
non-Markovian policies of the form given by \eqref{eq:policy}. Therefore, a stability notion for general stochastic processes is essential for our analysis of the system. 
In this paper, we call a stochastic process $\{X_t,t=0,1\dots\}$ \textit{stable} if
for every $\epsilon > 0$ there exists a finite set $K$ such that
\begin{align}
\mathbf{P}(X_t \notin K ) < \epsilon \text{ for all }t.
\label{def:stability}
\end{align}
This stability concept is also used in \cite{szpankowski1994stability,luo1999stability, borst2008stability}, and it is called bounded in probability in \cite{meyn2009markov}.
Note that the stability criterion \eqref{def:stability} is equivalent to positive recurrence for countable irreducible Markov chains \cite[Proposition 18.3.1]{meyn2009markov}. 
For general countable Markov chains with a reachable state, bounded in probability is equivalent to positive Harris recurrence, another stability concept for general Markov chains \cite[Proposition 18.3.2 ]{meyn2009markov}.


Given the arrival rates $\lambda = (\lambda^1,\dots,\lambda^N)$ to all queues, 
a policy/protocol $g$ stabilizes the communication system if the resulting queue length process
$\{Q^{n,g}_t, t=0,1,\dots\}$ for every user $n=1,\dots,N$ is stable.
The arrival rate $\lambda$ is said to be supportable if there exist policies/protocols that can stabilize the communication system under $\lambda$.

For any arrival rates $\lambda = (\lambda^1,\dots,\lambda^N)$, we use $\lambda^{tot}:=\sum_{n=1}^N \lambda^n$ to denote the total arrival rate to the communication system.
Since at most one packet can be transmitted through the collision channel at each time, only $\lambda \in \Lambda$ could be supportable, where
\begin{align}
\Lambda = \left\{\lambda = (\lambda^1,\lambda^2,\dots,\lambda^N): \lambda^{tot} < 1  \right\}.
\label{eq:throughputregion}
\end{align}

Furthermore, any $\lambda \in \Lambda$ is supportable by the time sharing policy that assigns $\lambda^n$ portion of time slots to user $n$. 
Therefore, arrival rates $\lambda$ are supportable if and only if $\lambda \in \Lambda$. We call $\Lambda$ the throughput region of the multiple access communication system.
We call a decentralized policy/protocol throughput optimal if it can stabilize the communication system for any $\lambda \in \Lambda$. 

\subsection{Queueing Delay} \label{sub:model:delay}
Let $Q^{tot}_t:= \sum_{n=1}^N Q^n_t$ denote total queue length of the system at time $t, t=1,2,\dots$.
We define
\begin{align}
Q_{avg}
:= \limsup_{t\rightarrow\infty} \frac{1}{T} \mathbf{E}\left[\sum_{t=0}^{T-1}Q^{tot}_t\right].
\end{align}
From Little's law (see \cite{bertsekas1992data}), 
in a stable queueing system, the queueing delay of a packet is proportional to the average total number of packets in the system. 
For a throughput optimal protocol $g$, the queueing delay of the system is given by $\frac{Q^g_{\text{avg}}}{\lambda^{tot}} $.

\subsection{Objective}
Our objective is to find a throughput optimal protocol that results in low queueing delay.

\section{The Common Information-Based Multiple Access (CIMA) Protocol}
\label{sec:CIMA}
\subsection{Preliminaries}
We first introduce common upper bounds for the queues.
Let $B^g_t := (B^{1,g}_t,B^{2,g}_t,\dots,B^{N,g}_t)$, where
$B^{n,g}_t$ is the upper bound on $Q^n_t$ at time slot $t$ based on the transmission protocol $g$ and the common information $F_{0:t-1}$, received from the common feedback, up to time slot $t$.
That is, when $F_{0:t-1}=f_{0:t-1}$,
\begin{align*}
b^{n,g}_t = 
&\max \{q^n_t: \exists\lambda\in \Lambda \text{ s.t. } \mathbf{P}^{\lambda,g}(q^n_t|f_{0:t-1}) >0\}.
\end{align*}
Note that, $B^g_t$ is a function of the common information $F_{0:t-1}$. We use $B^{g}_t$ to denote that the common upper bounds depend explicitly on the transmission policy $g$.

%
\subsection{The $\g$ Protocol}
\label{sub:ghat}
The $\g$ protocol is defined as follows.
\begin{align}
U^n_t = & \g^n_t(Q^n_{0:t}, A^n_{0:t-1}, U^n_{0:t-1}, F_{0:t-1}) \nonumber\\
      = & \left\{ \begin{array}{ll}
      1 & \text{ if } v(B^{\g}_t) = n \text{ and }Q^n_t >0 ,\\ 
      0 & \text{ otherwise,}
      \end{array}\right.
\label{eq:gmac}
\end{align}
where $v(\cdot)$ is a function of common upper bounds $B^{\g}_t$ defined as
\begin{align*}
v(b^{\g}_t) = \min\{n: b^{n,\g}_t = \max_{m=1,2,\dots,N}b^{m,\g}_t\}.
\end{align*}
Note that $v(b^{\g}_t)$ is the user with the largest common upper bound. Since we want to avoid collision, if there are more than one users with the largest common upper bound, $\g$ selects the user with the smallest index.


\section{Performance Analysis of the $\g$ Protocol}
\label{sec:analysis}
We prove that the $\g$ protocol is throughput optimal in \ref{sub:throughput}.
We provide an upper bound on the queueing delay under the $\g$ protocol in \ref{sub:delay}.

\subsection{Preliminary Results}
In order to analyze the system dynamics under the $\g$ protocol, we first provide the following result.
\begin{lemma}
\label{lm:conditional_indep}
Under the $\g$ protocol,
the queue lengths are independent conditional on the common feedback given any arrival rates $\lambda$.
Specifically, for any time $t$, any realization $f_{0:t-1}$ and any value $q_t = (q^1_t,\dots,q^N_t)$ of $Q_t = (Q^1_t,\dots,Q^N_t)$,
\begin{align}
\mathbf{P}^{\lambda,\g}(q_t|f_{0:t-1})
= & \prod_{n=1}^N \mathbf{P}^{\lambda,\g}(q^n_t|f_{0:t-1}).
\label{eq:conditional_indep}
\end{align}
Moreover, the conditional probability can be updated as follows.
For $n\neq v(b^{\g}_t)$
\begin{align}
&\mathbf{P}^{\lambda,\g}(q^n_{t+1}|f_{0:t})
\nonumber\\
=&\lambda^n\mathbf{P}^{\lambda,\g}(Q^n_t = q^n_{t+1}-1|f_{0:t-1})\nonumber\\
&+(1-\lambda^n)\mathbf{P}^{\lambda,\g}(Q^n_t = q^n_{t+1}|f_{0:t-1}).
\label{eq:conditional_update1}
\end{align}
For $n =  v(b^{\g}_t)$ and $f_t = 1$
\begin{align}
&\mathbf{P}^{\lambda,\g}(q^n_{t+1}|f_{0:t})
\nonumber\\
=&\lambda^n
\frac{\mathbf{P}^{\lambda,\g}(Q^n_t=q^n_{t+1}|f_{0:t-1})1_{\{q^n_{t+1}>0\}}}{\mathbf{P}^{\lambda,\g}(Q^n_t >0|f_{0:t-1})}
\nonumber\\
&+(1-\lambda^n)
\frac{\mathbf{P}^{\lambda,\g}(Q^n_t=q^n_{t+1}+1|f_{0:t-1})}{\mathbf{P}^{\lambda,\g}(Q^n_t >0|f_{0:t-1})}.
\label{eq:conditional_update2}
\end{align}
For $n =  v(b^{\g}_t)$ and $f_t = 0$
\begin{align}
\mathbf{P}^{\lambda,\g}(q^n_{t+1}|f_{0:t})
=& \left\{ \begin{array}{ll}
0 & \text{ if }q^n_{t+1} \geq 2, \\
\lambda^n & \text{ if }q^n_{t+1} =1, \\
1-\lambda^n & \text{ if }q^n_{t+1} =0.
\end{array}\right.
\label{eq:conditional_update3}
\end{align}

\end{lemma}
\begin{proof}
See Appendix \ref{app:lm1}.
\end{proof}

Using Lemma \ref{lm:conditional_indep}, we can obtain the evolution of queue lengths and common upper bounds under $\g$, as stated in the lemma below.
\begin{lemma}
\label{lm:evolution}
Under $\g$, the queue lengths evolve as
\begin{align}
Q^{n,\g}_{t+1}
 = & \left\{ \begin{array}{ll}
      A^n_t+Q^{n,\g}_{t} & \text{ if } n \neq v(B^{\g}_t) ,\\ 
      A^n_t+\left(Q^{n,\g}_{t} - 1\right)^+ & \text{ if } n = v(B^{\g}_t).
      \end{array}\right.
\label{eq:Qunderg}
\end{align}
and the common upper bounds evolve according to
\begin{align}
&B^{n,\g}_{t+1}\nonumber\\
 = & \left\{ \begin{array}{ll}
      B^{n,\g}_{t}+1 & \text{ if } n \neq v(B^{\g}_t), \\ 
      B^{n,\g}_{t}   & \text{ if } n  = v(B^{\g}_t) \text{ and } F_t=1, \\
      1   & \text{ if } n = v(B^{\g}_t) \text{ and } F_t =0 .
      \end{array}\right.
\label{eq:Bunderg}
\end{align}
\end{lemma}
\begin{proof}
See Appendix \ref{app:lm2}.
\end{proof}
Using Lemma \ref{lm:evolution},
the $\g$ protocol can be easily implemented as described in Algorithm \ref{alg:g} below.
\begin{algorithm}[H]
\caption{The $\g$ protocol for user $n \in\{1,2,\dots,N\}$}
\label{alg:g}
\begin{algorithmic}
\FOR{$k=1$ to $N$}
\STATE $B^k \leftarrow 0$
\ENDFOR

\WHILE{user $n$ is active}
\STATE $B^{\text{MAX}} \leftarrow \max_k(B^k)$
\STATE $v \leftarrow \min(k: B^k = B^{\text{MAX}})$
\IF{$n = v$ and $Q^n_t >0$}
\STATE transmit a packet (set $U_t = 1$)
\ENDIF

\FOR{$k \neq v$}
\STATE $B^k \leftarrow B^k+1$
\ENDFOR

\IF{$F_t = 1$}
\STATE  $B^v \leftarrow B^v$
\ELSE
\STATE  $B^v \leftarrow 1$
\ENDIF
\ENDWHILE
\end{algorithmic}
\end{algorithm}

\subsection{Throughput Optimality} \label{sub:throughput}
The main result on $\g$'s throughput is stated in the following theorem.
\begin{theorem}
\label{thm:throughput_optimal}
The $\g$ protocol is throughput optimal.
That is, for any arrival rates $\lambda \in \Lambda$ (defined by \eqref{eq:throughputregion}),
the queue length processes under $\g$ are stable.
\end{theorem}

To prove the theorem, we first show that under the $\g$ protocol the queue lengths together with the upper bounds form a Markov chain.
\begin{lemma}
\label{lm:MC}
Let $Y^{\g}_t 
:= (Q^{\g}_t, B^{\g}_t)$, where
\begin{align*}
Q^{\g}_t = &(Q^{1,\g}_t,Q^{2,\g}_t,\dots,Q^{N,\g}_t)
\end{align*}
for every time slot $t=0,1,\dots$
Then, $\{Y^{\g}_t, t=0,1,\dots\}$ is a Markov chain.
\end{lemma}
\begin{proof}
See Appendix \ref{app:lm3}.
\end{proof}

Since $\{Y^{\g}_t, t=0,1,\dots\}$ is a Markov chain, we can use the Foster-Lyapunov theorem in the proof below to show that the process $\{Y^{\g}_t, t=0,1,\dots\}$ is stable.

\begin{proof}[Proof of Theorem \ref{thm:throughput_optimal}]
Let $\epsilon = 1 - \lambda^{tot}$. Then $\epsilon>0$ because $\lambda \in \Lambda$.
Let $y:=(q,b)=(q^1,q^2,\dots,q^N,b^1,b^2,\dots,b^N)$.
Define the Lyapunov function $h(y)$ by
\begin{align}
h(y) = \sum_{n=1}^N (q^n + \alpha b^n),
\label{eq:lyafun}
\end{align}
where $\alpha = \frac{\epsilon}{2(N-1)}$.
For $Y^{\g}_t = y$,
let $v = v(b)=\min(n: b^n=\max_{k\in \mathcal{N}}(b^{k}))$.
Then from \eqref{eq:Qunderg} and \eqref{eq:Bunderg} in Lemma \ref{lm:evolution} we get
\begin{align}
 & \mathbf{E}\left[h(Y^{\g}_{t+1}) - h(Y^{\g}_t) | Y^{\g}_t=y\right] \nonumber\\
\leq & -\epsilon/2 \quad \text{   if } b^v \geq \frac{1}{\alpha}+1.
\label{eq:stability1}
\end{align}
(see Appendix \ref{app:thm} for a detailed derivation of \eqref{eq:stability1})
\\
Since $b^v=\max_{k \in \mathcal{N}}(b^{k})$, $b^v \geq b^n$ and $b^v\geq q^n$ for all $n=1,2,\dots,N$.
Define 
\begin{align*}
C = \{y=(q,b): q^n<\frac{1}{\alpha}+1, b^n<\frac{1}{\alpha}+1 \quad\forall n\}.
\end{align*}
Then, \eqref{eq:stability1} holds for every $y \notin C$.
Since $C$ is a finite set, the Foster-Lyapunov drift criterion (Condition (DD2) in \cite{meyn1992stability}) is satisfied.
From \cite[Theorem 4.5]{meyn1992stability}, $\{Y^{\g}_t, t=0,1,\dots\}$ is bounded in probability (satisfies the stability condition \eqref{def:stability}).

Therefore, for every $\epsilon>0$ there exists a finite set $K$ such that
\begin{align}
\mathbf{P}(Y^{\g}_t \notin K ) < \epsilon \text{ for all }t.
\label{eq:bddinprob}
\end{align}
Let $K^n = \{q^n: \text{there exists }y = (q,b) \in K\}$ be the projection of $K$ on its $n$th component.
Then,
\begin{align}
\mathbf{P}(Q^{n,\g}_t \notin K^n )
\leq 
& \mathbf{P}(Y^{\g}_t \notin K ) < \epsilon \text{ for all }t.
\label{eq:bddinprob2}
\end{align}
Therefore,
$\{Q^{n,\g}_t,t=0,1,\dots\}$
also satisfies \eqref{def:stability} and the stability of the communication system under $\g$ is established.


\end{proof}
\begin{remark}
We provide an alternative proof of Theorem \ref{thm:throughput_optimal}.
\\
As a result of \eqref{eq:stability1}, condition (V2) in \cite[Chap. 11]{meyn2009markov} is satisfied.
Therefore, by Theorem 11.3.4 in \cite{meyn2009markov} the Markov chain $\{Y^{\g}_t, t=0,1,\dots\}$ is positive Harris recurrent on a countable state space. By Theorem 18.3.2 in \cite{meyn2009markov} positive Harris recurrence implies \eqref{eq:bddinprob}, which in turn implies \eqref{eq:bddinprob2}, and this establishes the assertion of Theorem \ref{thm:throughput_optimal}.

\end{remark}
\subsection{Delay Performance} \label{sub:delay}

Using $\g$, we have the following queueing delay performance guarantee.
\begin{theorem}
\label{thm:delay}
Under the $\g$ protocol, for any rate $\lambda \in \Lambda$ we have
\begin{align}
\frac{Q^{\g}_{avg}}{\lambda^{tot}} \leq \frac{2 N}{1-\lambda^{tot}}.
\label{eq:delaybound}
\end{align}
\end{theorem}

Theorem \ref{thm:delay} says that for any fixed total arrival rate $\lambda^{tot}$, the queuing delay under the $\g$ protocol is linear in the number of users $N$. 

To prove Theorem \ref{thm:delay}, we first present a property of the $\g$ protocol.
\begin{lemma}
\label{lm:successfulq}
Let $\bar{U}_t = \sum_{n=1}^N U^n_t\prod_{m\neq n}(1-U^m_t)$.
If the total number of packets at time $t$ is $Q^{tot,\g}_{t} = q$, there are at least $q$ successful transmissions from time $t$ to $t+q+N-1$ using the $\g$ protocol. That is
\begin{align*}
\sum_{\tau=t}^{t+q+N-1} \bar{U}^{\g}_{\tau} \geq q.
\end{align*}

\end{lemma}
\begin{proof}
See Appendix \ref{app:lm4}.
\end{proof}

Since $U^n_t\in\{0,1\}$ for each $n,n=1,2,\dots,N$, $\bar{U}_t\in \{0,1\}$; $\bar{U}_t=1$ (respectively, $\bar{U}_t=0$) denotes a successful (respectively, unsuccessful) transmission at time $t$.
Lemma \ref{lm:successfulq} shows that when at a certain time slot the total queue length is $q$,
the $\g$ protocol can successfully transmit at least $q$ packets in the next $q+N-1$ time slots.

Using Lemma \ref{lm:successfulq}, we can now prove Theorem \ref{thm:delay}.

\begin{proof}[Proof of Theorem \ref{thm:delay}]
Let $T_1 = N$, and define recursively the random variables $T_2,T_3,\dots$ by
\begin{align*}
T_k = \left\{\min t: t>T_{k-1}, \sum_{\tau = T_{k-1}}^{t-1} \bar{U}^\g_{\tau} = Q^{tot,\g}_{T_{k-1}}  \right\}.
\end{align*}
Then, each $T_k$ is the time such that $Q^{tot,\g}_{T_{k-1}}$ packets are successfully transmitted
from time $T_{k-1}$ to $T_{k}-1$ under the $\g$ protocol.

By Lemma \ref{lm:successfulq} the $\g$ protocol can successfully transmit at least $Q^{tot,\g}_{T_{k-1}}$ packets from time
$T_{k-1}$ to $T_{k-1}+Q^{tot,\g}_{T_{k-1}}+N-1$. Therefore
\begin{align}
T_k \leq T_{k-1}+Q^{tot,\g}_{T_{k-1}}+N.
\label{eq:Tk}
\end{align}
Consequently, from the dynamics of queues and \eqref{eq:Tk} we obtain
\begin{align}
 \mathbf{E}\left[Q^{tot,\g}_{T_k}\right] 
\leq & \lambda^{tot}\left(\mathbf{E}\left[Q^{tot,\g}_{T_{k-1}}\right]+N\right).
\label{eq:QTkrec}
\end{align}
(see Appendix \ref{app:thm:delay} for a detailed derivation of \eqref{eq:QTkrec})
\\
Since $T_1 = N$, $\mathbf{E}\left[Q^{tot,\g}_{T_{1}}\right] 
\leq \mathbf{E}\left[\sum_{t=0}^{N-1} \sum_{n=1}^N A^n_t\right] = \lambda^{tot} N $.
From \eqref{eq:QTkrec}, we can show, recursively, that for all $k$
\begin{align}
 \mathbf{E}\left[Q^{tot,\g}_{T_k}\right] \leq &\lambda^{tot} N + (\lambda^{tot})^2 N+\dots+(\lambda^{tot})^k N \nonumber\\
                                          \leq &\frac{\lambda^{tot} N}{1-\lambda^{tot}}.
\label{eq:QTkbound}
\end{align}
Now for any time $t=0,1,2,\dots$, for any realization of arrivals there is some number $k$ such that $ T_{k-1} < t \leq T_{k} $ ($T_0:= 0$). Using \eqref{eq:QTkbound} and the dynamics of queues we get
\begin{align}
\mathbf{E}\left[Q^{tot,\g}_{t}\right] 
\leq & 2 \frac{\lambda^{tot} N}{1-\lambda^{tot}}
\label{eq:qt}
\end{align}
(see Appendix \ref{app:thm:delay} for a detailed derivation of \eqref{eq:qt})
\\
Since \eqref{eq:qt} holds for any time $t$, we have
\begin{align}
Q^\g_{\text{avg}}
= &\limsup_{t\rightarrow\infty} \frac{1}{T} \sum_{t=0}^{T-1}\mathbf{E}\left[Q^{tot,\g}_t\right] \nonumber\\
\leq & \limsup_{t\rightarrow\infty} \frac{1}{T} \sum_{t=0}^{T-1}\frac{2\lambda^{tot} N}{1-\lambda^{tot}} \nonumber\\
= & \frac{2\lambda^{tot} N}{1-\lambda^{tot}}.
\end{align}

\end{proof}

\begin{remark}
The result of Theorem \ref{thm:delay} implies throughput optimality of the $\g$ protocol.
Since the bound on delay (the right hand side of \eqref{eq:delaybound}) is finite for every $\lambda\in\Lambda$, it can be shown that the stability requirement described by \eqref{def:stability} is satisfied. Nevertheless, the proof of Theorem \ref{thm:throughput_optimal} is interesting/instructive by itself, and for this reason we have proved throughput optimality and delay performance separately.
\end{remark}

\section{Simulation Results}
\label{sec:simulation}

In this section we first compare, via simulation, the queueing delay incurred by $\g$ with 
that of three other protocols that use the same feedback information and no channel sensing: the basic TDMA protocol, the adaptive TDMA (ATDMA) protocol \cite{papadimitriou2003adaptive} and the quadratic back-off protocol which is proved to be throughput optimal in \cite{haastad1996analysis}. 
In the quadratic back-off protocol, each user transmits a packet with probability $(c+1)^{-2}$ where $c$ is the back-off counter.
We also compare the delay performance of $\g$ with CSMA protocols proposed in\cite{chlebus2012adversarial,anantharamu2009adversarial}; these protocols employ channel sensing before transmission scheduling.

In the numerical experiments, we have used different values of $N$ and $\lambda^{tot}$ for each protocol.
Arrival rates are asymmetric: half of the users have arrival rate $1.4\lambda^{tot}/N$ and the other half of the users have arrival rate $0.6\lambda^{tot}/N$.
For each $N$ and $\lambda^{tot}$, we run the simulation for $T = 10^5$ time steps. 


The simulation results of Fig. \ref{fig:CIMA} show that the average delay associated with the $\g$ protocol is linear in the number of users. These simulation results are consistent with the result of Theorem \ref{thm:delay}.

In Fig. \ref{fig:Compare}, we compare the delay performance of TDMA, ATDMA, quadratic back-off and CIMA protocols for a system of $4$ users.
Fig. \ref{fig:Compare} shows that the delay associated with the $\g$ protocol is significantly smaller than that of the quadratic back-off protocol (that is also throughput optimal) and of the TDMA protocol (note that TDMA is unstable when $\lambda^{tot}>0.7$ ). $\g$'s delay is also smaller than the delay of ATDMA (note that there is no theoretical analysis for ATDMA).

In Fig. \ref{fig:Compare_CSMA}, we compare the delay performance of $\g$ with two CSMA protocols: 
the PGD protocol proposed in \cite{jiang2012fast}, and the DCSMA protocol proposed in \cite{lee2014provable}.
The results of \cite{jiang2012fast} and \cite{lee2014provable} prove that the two CSMA protocols achieve delay that is linear in the number of system users.
That is, the delay of the CSMA protocols is of the same order as the delay of the\ $\g$ protocol.
However, channel sensing is required to implement the two CSMA protocols.
Moreover, Fig. \ref{fig:Compare_CSMA} shows that the delay resulting from the $\g$ protocol is significantly smaller than that of the CSMA protocols of \cite{jiang2012fast} and \cite{lee2014provable}.

\begin{figure}
\includegraphics[width=0.5 \textwidth]{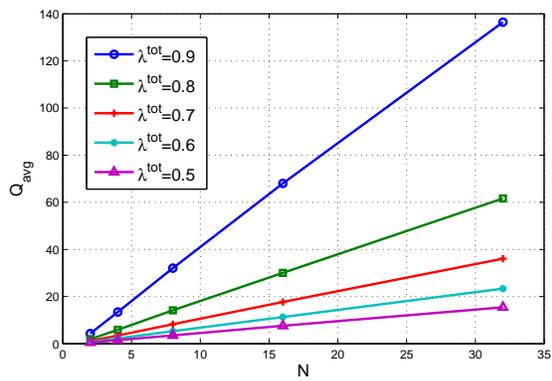}
\caption{Delay versus the number of users of $\g$}
\label{fig:CIMA}
\end{figure}
\begin{figure}
\includegraphics[width=0.5 \textwidth]{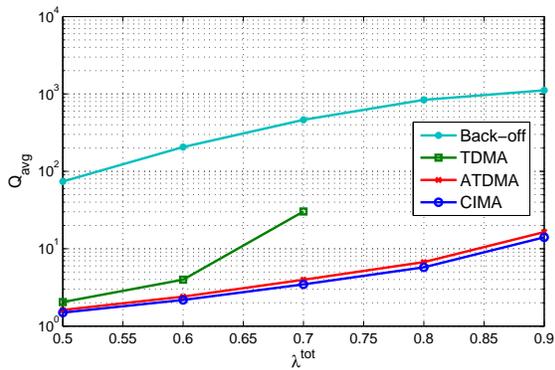}
\caption{Comparison of protocols for a system of $4$ users}
\label{fig:Compare}
\end{figure}

\begin{figure}
\includegraphics[width=0.5 \textwidth]{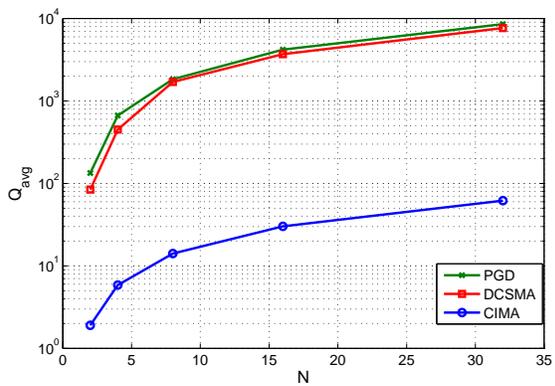}
\caption{Comparison of $\g$ and CSMA protocols}
\label{fig:Compare_CSMA}
\end{figure}

\section{Conclusion}
\label{sec:conclusion}
We developed a transmission protocol that utilizes the common information of the system's users to achieve efficient/optimal coordination of their transmissions.
The protocol is collisions free; thus, it is similar in spirit to TDMA (or adaptive TDMA), but it differs from TDMA in the way it selects the user to transmit at each time slot.
Intuitively, we expect that the delay due to the $\g$ protocol will increase linearly with the number of users. The result of Theorem \ref{thm:delay} confirms this intuition.

The problem investigated in this paper can be viewed as a decentralized control/decision-making problem with non-classical information structure \cite{witsenhausen1971separation}. Decentralized stochastic control problems with non-classical information structure are signaling problems \cite{ho1980team}. In our setup signaling occurs through the feedback provided by the collision channel. Signaling leads to adjustments of each user's upper bounds on their queue lengths (in the manner described by $\g$) and results in efficient coordination among the users.
Signaling also occurs in CSMA protocols and in adaptive TDMA, but it is distinctly different from that of the $\g$ protocol.


%

\appendices
\section{}
\label{app:lm1}
\begin{proof}[Proof of Lemma \ref{lm:conditional_indep}]
The lemma is proved by induction.
Equation \eqref{eq:conditional_indep} is true at $t=0$ because all queues are initially empty.
Suppose \eqref{eq:conditional_indep} is true at $t=k$. At time $t=k+1$ we have
\begin{align}
&\mathbf{P}^{\lambda,\g}(q_{k+1}|f_{0:k})\nonumber\\
= &
\frac{\mathbf{P}^{\lambda,\g}(q_{k+1}, f_k|f_{0:k-1})}{\sum_{q'_{k+1}}\mathbf{P}^{\lambda,\g}(q'_{k+1},f_k|f_{0:k-1}) }.
\label{eq:lm_conditional1}
\end{align}
Let $v = v(b^{\g}_t)$.
Consider the numerator in \eqref{eq:lm_conditional1}. There are two cases: $f_k = 1$ and $f_k = 0$.
\\
When $f_k = 1$, we have
\begin{align}
&\mathbf{P}^{\lambda,\g}(q_{k+1}, F_k=1|f_{0:k-1})\nonumber\\
= &\sum_{q_k}
\mathbf{P}^{\lambda,\g}(q_{k+1},q_{k}, F_k=1|f_{0:k-1})\nonumber\\
= &\sum_{q_k}
\mathbf{P}^{\lambda,\g}(q_{k+1},q_{k}, Q^{v}_{k}>0|f_{0:k-1})\label{eq:lm_conditional2_p1}\\
= &\sum_{q_k, q^{v}_{k}>0}
\left[\vphantom{\prod_{n\neq v}} \mathbf{P}^{\lambda,\g}(q_{k}|f_{0:k-1})\right.\nonumber\\
&\left.
\mathbf{P}^{\lambda}(A^v_k = q^v_{k+1}-q^v_{k}+1)\prod_{n\neq v}\mathbf{P}^{\lambda}(A^n_k = q^n_{k+1}-q^n_{k}) \right]
\label{eq:lm_conditional2_p2}\\
= &\sum_{q_k, q^{v}_{k}>0}
\left[\vphantom{\prod_{n\neq v}} \prod_{n =1}^N \mathbf{P}^{\lambda,\g}(q^n_{k}|f_{0:k-1})\right.\nonumber\\
&\left.
\mathbf{P}^{\lambda}(A^v_k = q^v_{k+1}-q^v_{k}+1)\prod_{n\neq v}\mathbf{P}^{\lambda}(A^n_k = q^n_{k+1}-q^n_{k}) \right]
\label{eq:lm_conditional2_p3}\\
= 
&\prod_{n\neq v}\left[\sum_{q^n_k}\mathbf{P}^{\lambda}(A^n_k = q^n_{k+1}-q^n_{k})\mathbf{P}^{\lambda,\g}(q^n_{k}|f_{0:k-1}) \right]
\nonumber\\
&
\sum_{q^{v}_{k}>0}\mathbf{P}^{\lambda}(A^v_k = q^v_{k+1}-q^v_{k}+1)
\mathbf{P}^{\lambda,\g}(q^v_{k}|f_{0:k-1}).
\label{eq:lm_conditional2}
\end{align}
Equation \eqref{eq:lm_conditional2_p1} holds because $F_k = 1$ if and only if $Q^{v(B^{\g}_t)}_{k}>0$.
Equation \eqref{eq:lm_conditional2_p2} is true because of the system dynamics \eqref{eq:qdynamics} and the fact that
$A^n_k, n=1,2,\dots,N$ are mutually independent and independent of all variables before $k$. 
Equation \eqref{eq:lm_conditional2_p3} follows from the induction hypothesis for \eqref{eq:conditional_indep}. 
Equation \eqref{eq:lm_conditional2} is true because each term in \eqref{eq:lm_conditional2_p3} depends only on each $q^n_k$ for $n=1,2,\dots,N$.
\\
Using \eqref{eq:lm_conditional2}, the denominator in \eqref{eq:lm_conditional1} becomes
\begin{align}
&\sum_{q'_{k+1}}\left\{
\prod_{n\neq v}\left[\sum_{q^n_k}\mathbf{P}^{\lambda}(A^n_k = q'^{n}_{k+1}-q^n_{k})\mathbf{P}^{\lambda,\g}(q^n_{k}|f_{0:k-1}) \right]
\right.
\nonumber\\
&\quad\left.
\sum_{q^{v}_{k}>0}\mathbf{P}^{\lambda}(A^v_k = q'^{v}_{k+1}-q^v_{k}+1)
\mathbf{P}^{\lambda,\g}(q^v_{k}|f_{0:k-1})\right\}
\nonumber\\
=&
\sum_{q^{v}_{k}>0}\mathbf{P}^{\lambda,\g}(q^v_{k}|f_{0:k-1})\nonumber\\
=&
\mathbf{P}^{\lambda,\g}(Q^v_{k}>0|f_{0:k-1});
\label{eq:lm_conditional2d}
\end{align}
equation \eqref{eq:lm_conditional2d} is true because all possible values of $q^n_k,n\neq v$ are summed out.
\\
Substituting \eqref{eq:lm_conditional2} and \eqref{eq:lm_conditional2d} back into \eqref{eq:lm_conditional1} we obtain for $f_k=1$
\begin{align}
&\mathbf{P}^{\lambda,\g}(q_{k+1}|f_{0:k})\nonumber\\
=
&\prod_{n\neq v}\left[\sum_{q^n_k}\mathbf{P}^{\lambda}(A^n_k = q^n_{k+1}-q^n_{k})\mathbf{P}^{\lambda,\g}(q^n_{k}|f_{0:k-1}) \right]
\nonumber\\
&
\left[\sum_{q^{v}_{k}>0}\mathbf{P}^{\lambda}(A^v_k = q^v_{k+1}-q^v_{k}+1)
\frac{\mathbf{P}^{\lambda,\g}(q^v_{k}|f_{0:k-1})}{\mathbf{P}^{\lambda,\g}(Q^v_{k}>0|f_{0:k-1})}\right]
\nonumber\\
=: & \prod_{n=1}^N \phi^n(q^n_{k+1}),
\label{eq:lm_conditional2r}
\end{align}
where, for $n\neq v$,
\begin{align}
&\phi^n(q^n_{k+1})\nonumber\\
:=&
\sum_{q^n_k}\mathbf{P}^{\lambda}(A^n_k = q^n_{k+1}-q^n_{k})\mathbf{P}^{\lambda,\g}(q^n_{k}|f_{0:k-1}) \nonumber\\
=&\mathbf{P}^{\lambda}(A^n_k = 1)\mathbf{P}^{\lambda,\g}(Q^n_k = q^n_{k+1}-1|f_{0:k-1})
\nonumber\\
&+\mathbf{P}^{\lambda}(A^n_k = 0)\mathbf{P}^{\lambda,\g}(Q^n_k = q^n_{k+1}|f_{0:k-1})
\nonumber\\
=&\lambda^n\mathbf{P}^{\lambda,\g}(Q^n_k = q^n_{k+1}-1|f_{0:k-1})
\nonumber\\
&+(1-\lambda^n)\mathbf{P}^{\lambda,\g}(Q^n_k = q^n_{k+1}|f_{0:k-1}),
\label{eq:lm_conditional2r_n}
\end{align}
and for $n=v$,
\begin{align}
&\phi^v(q^v_{k+1})\nonumber\\
:= &
\sum_{q^{v}_{k}>0}\mathbf{P}^{\lambda}(A^v_k = q^v_{k+1}-q^v_{k}+1)
\frac{\mathbf{P}^{\lambda,\g}(q^v_{k}|f_{0:k-1})}{\mathbf{P}^{\lambda,\g}(Q^v_{k}>0|f_{0:k-1})} \nonumber\\
= & \mathbf{P}^{\lambda}(A^v_k = 1)
\frac{\mathbf{P}^{\lambda,\g}(Q^v_{k} = q^v_{k+1}|f_{0:k-1})\mathbf{1}_{\{q^n_{t+1}>0\}}}{\mathbf{P}^{\lambda,\g}(Q^v_{k}>0|f_{0:k-1})}
\nonumber\\
&+ \mathbf{P}^{\lambda}(A^v_k = 0)
\frac{\mathbf{P}^{\lambda,\g}(Q^v_{k} = q^v_{k+1}+1|f_{0:k-1})}{\mathbf{P}^{\lambda,\g}(Q^v_{k}>0|f_{0:k-1})}
\nonumber\\
= & \lambda^v
\frac{\mathbf{P}^{\lambda,\g}(Q^v_{k} = q^v_{k+1}|f_{0:k-1})\mathbf{1}_{\{q^v_{k+1}>0\}}}{\mathbf{P}^{\lambda,\g}(Q^v_{k}>0|f_{0:k-1})}
\nonumber\\
&+ (1-\lambda^v)
\frac{\mathbf{P}^{\lambda,\g}(Q^v_{k} = q^v_{k+1}+1|f_{0:k-1})}{\mathbf{P}^{\lambda,\g}(Q^v_{k}>0|f_{0:k-1})}.
\label{eq:lm_conditional2r_v}
\end{align}
Equations \eqref{eq:lm_conditional2r_n} and \eqref{eq:lm_conditional2r_v} follow from \eqref{eq:lm_conditional2r} and the fact that $A^n_k$ takes values in $\{0,1\}$ for all $n=1,2,\dots,N$.
\\
From \eqref{eq:lm_conditional2r_n} and \eqref{eq:lm_conditional2r_v} we conclude that $\phi^n(q^n_{k+1})$ is a probability mass function (PMF) for all $n$. This along with \eqref{eq:lm_conditional2r} establish that the marginal conditional PMF satisfies
\begin{align}
\mathbf{P}^{\lambda,\g}(q^n_{k+1}|f_{0:k}) 
=&\phi^n(q^n_{k+1})
\label{eq:lm_conditional_marginal}
\end{align}
for all $n$, and this establish the induction step when $f_k=1$.
\\
When $f_k = 0$, by arguments similar to those in \eqref{eq:lm_conditional2_p1}-\eqref{eq:lm_conditional_marginal}, we get
\begin{align}
\mathbf{P}^{\lambda,\g}(q_{k+1}|f_{0:k})
= & \prod_{n=1}^N \mathbf{P}^{\lambda,\g}(q^n_{k+1}|f_{0:k}),
\end{align}
where for $n\neq v$, by an argument similar to \eqref{eq:lm_conditional2r_n}, we get
\begin{align}
&\mathbf{P}^{\lambda,\g}(q^n_{k+1}|f_{0:k}) \nonumber\\
=&\lambda^n\mathbf{P}^{\lambda,\g}(Q^n_k = q^n_{k+1}-1|f_{0:k-1})
\nonumber\\
&+(1-\lambda^n)\mathbf{P}^{\lambda,\g}(Q^n_k = q^n_{k+1}|f_{0:k-1}),
\label{eq:lm_conditional2r_n0}
\end{align}
and for $n=v$, by an argument similar to \eqref{eq:lm_conditional2r_v}, we obtain
\begin{align}
&\mathbf{P}^{\lambda,\g}(q^v_{k+1}|f_{0:k})  \nonumber\\
= & \lambda^v \mathbf{1}_{\{q^v_{k+1}=1\}}+ (1-\lambda^v)\mathbf{1}_{\{q^v_{k+1}=0\}}.
\label{eq:lm_conditional2r_v0}
\end{align}
Therefore, the induction step is complete and \eqref{eq:conditional_indep} holds for all $t$.
Furthermore, \eqref{eq:conditional_update1} is established by \eqref{eq:lm_conditional2r_n} and \eqref{eq:lm_conditional2r_n0};
\eqref{eq:conditional_update2} is established by \eqref{eq:lm_conditional2r_v}, and \eqref{eq:conditional_update3} is established by \eqref{eq:lm_conditional2r_v0}.

\end{proof}

\section{}
\label{app:lm2}
\begin{proof}[Proof of Lemma \ref{lm:evolution}]
Equation \eqref{eq:Qunderg} follows directly from \eqref{eq:qdynamics}, the queue length dynamics, and \eqref{eq:gmac}, the definition of the $\g$ protocol.
\\
For the common upper bounds, 
let $v = v(B^{\g}_t)$, which is a function of $F_{0:t-1}$.
\\
For $n \neq v$,  we get $B^{n,\g}_{t+1} = B^{n,\g}_{t}+1$ form \eqref{eq:conditional_update1} in Lemma \ref{lm:conditional_indep}.
\\
For $n = v$ and $F_{t} =1$, we obtain
$B^{v,\g}_{t+1} = B^{v,\g}_{t}$ from \eqref{eq:conditional_update2} in Lemma \ref{lm:conditional_indep}.
\\
For $n = v$ and $F_t = 0$, \eqref{eq:conditional_update3} in Lemma \ref{lm:conditional_indep} gives $B^{v,\g}_{t+1} = 1$, and the proof of the lemma is complete.

\end{proof}

\section{}
\label{app:lm3}
\begin{proof}[Proof of Lemma \ref{lm:MC}]
From \eqref{eq:Qunderg} and \eqref{eq:Bunderg} in Lemma \ref{lm:evolution} we know that
$Q^{\g}_{t+1}$ and $B^{\g}_{t+1}$ are functions of 
$Q^{\g}_{t},B^{\g}_{t}$, $A^n_t$ and $F_t$.
From \eqref{eq:gmac}, the definition of the $\g$ protocol, we know that
\begin{align*}
F_t = U^{v(B^{\g}_t)}_t
 = \mathbf{1}_{\left\{Q^{v(B^{\g}_t),\g}_{t}>0\right\}}.
\end{align*}
Therefore, $F_t$ is a function of $Q^{\g}_{t}$ and $B^{\g}_{t}$.
Consequently, $Y^{\g}_{t+1}$ is a function of $Q^{\g}_{t},B^{\g}_{t}$ and $A^n_t$.
Let $f(Y^{\g}_t,A^n_t):=Y^{\g}_{t+1}$, we have
\begin{align*}
&\mathbf{P}(
Y^{\g}_{t+1} = y_{t+1} 
|Y^{\g}_k = y_k, k\leq t)\\
= & \mathbf{P}(
f(Y^{\g}_t,A^n_t) = y_{t+1} 
|Y^{\g}_k = y_k, k\leq t)\\
= & \mathbf{P}(
f(y_t,A^n_t) = y_{t+1} 
|Y^{\g}_k = y_k, k\leq t)\\
\stackrel{(*)}{=}  & \mathbf{P}(
f(y_t,A^n_t) = y_{t+1} 
|Y^{\g}_t = y_t)\\
=&\mathbf{P}(
Y^{\g}_t = y_{t+1} 
|Y^{\g}_t = y_t),
\end{align*}
where (*) is true because $A^n_t$ is independent of $Q^{\g}_{t},B^{\g}_{t}$ and all random variables before time slot $t$.
\\
Therefore, $\{Y^{\g}_t, t=0,1,\dots\}$ is a Markov chain.

\end{proof}

\section{}
\label{app:thm}
\begin{proof}[Detailed derivation of \eqref{eq:stability1} in the proof of Theorem \ref{thm:throughput_optimal}]
From the definition of the Lyapunov function $h(\cdot)$ ( cf \eqref{eq:lyafun}) we have
\begin{align}
 & \mathbf{E}\left[h(Y^{\g}_{t+1}) - h(Y^{\g}_t) | Y^{\g}_t=y\right] \nonumber\\
= &\mathbf{E}\left[ 
\sum_{n=1}^N (Q^{n,\g}_{t+1}-Q^{n,\g}_t) | Y^{\g}_t=y\right]\nonumber\\
& + \alpha \mathbf{E}\left[ 
\sum_{n=1}^N (B^{n,\g}_{t+1}- B^{n,\g}_t) | Y^{\g}_t=y\right].
\label{eq:app:thm:eq1}
\end{align}
For $n\neq v$, from \eqref{eq:Qunderg} and \eqref{eq:Bunderg} in Lemma \ref{lm:evolution}, we get
\begin{align}
 &\mathbf{E}\left[ (Q^{n,\g}_{t+1}-Q^{n,\g}_t) | Y^{\g}_t=y\right]\nonumber\\
& + \alpha \mathbf{E}\left[ 
(B^{n,\g}_{t+1}- B^{n,\g}_t) | Y^{\g}_t=y\right]\nonumber\\
= &\mathbf{E}\left[ A^n_t| Y^{\g}_t=y\right] + \alpha \mathbf{E}\left[ 1 | Y^{\g}_t=y\right]\nonumber\\
= & \lambda^n + \alpha,
\label{eq:app:thm:eq1n}
\end{align}
where the last equality in \eqref{eq:app:thm:eq1n} holds because $A^n_t$ is independent of $Y^{\g}_t$.
\\
For $n= v$, from \eqref{eq:Qunderg} and \eqref{eq:Bunderg} in Lemma \ref{lm:evolution}, we obtain
\begin{align}
 &\mathbf{E}\left[ (Q^{v,\g}_{t+1}-Q^{v,\g}_t) | Y^{\g}_t=y\right]\nonumber\\
& + \alpha \mathbf{E}\left[ 
(B^{v,\g}_{t+1}- B^{v,\g}_t) | Y^{\g}_t=y\right]\nonumber\\
= &\mathbf{E}\left[ A^v_t-1 + 1_{\{Q^v_t=0\}}| Y^{\g}_t=y\right] \nonumber\\
&+ \alpha \mathbf{E}\left[ (1-B^v_t)1_{\{Q^v=0\}} | Y^{\g}_t=y\right]\nonumber\\
= &\mathbf{E}\left[ A^v_t-1 + 1_{\{q^v=0\}}| Y^{\g}_t=y\right] \nonumber\\
&+ \alpha \mathbf{E}\left[ (1-b^v)1_{\{q^v=0\}} | Y^{\g}_t=y\right]\nonumber\\
= & \lambda^v - 1 + (1 + \alpha(1-b^v))   1_{\{q^v=0\}},
\label{eq:app:thm:eq1v}
\end{align}
where the last equality in \eqref{eq:app:thm:eq1v} follows from the fact that $A^v_t$ is also independent of $Y^{\g}_t$.
\\
substituting \eqref{eq:app:thm:eq1n} and \eqref{eq:app:thm:eq1v} back into \eqref{eq:app:thm:eq1} we get
\begin{align}
 & \mathbf{E}\left[h(Y^{\g}_{t+1}) - h(Y^{\g}_t) | Y^{\g}_t=y\right] \nonumber\\
= &\sum_{n \neq v} \lambda^n +\alpha(N-1) \nonumber\\
&+ \lambda^v - 1 + (1 + \alpha(1-b^v))   1_{\{q^v=0\}}\nonumber\\
\stackrel{(a)}{=}  &-\epsilon +\alpha(N-1)
+ (1+\alpha(1-b^v))1_{\{q^v=0\}} \nonumber\\
\stackrel{(b)}{=}  &-\epsilon/2
+ (1+\alpha(1-b^v))1_{\{q^v=0\}} \nonumber\\
\leq & -\epsilon/2 \quad \text{   if } b^v \geq \frac{1}{\alpha}+1,
\label{eq:app:thm:eqlast}
\end{align}
where (a) in \eqref{eq:app:thm:eqlast} is true because $\sum_{n=1}^N \lambda^n = 1-\epsilon$, and (b) in \eqref{eq:app:thm:eqlast} is true because $\alpha = \frac{\epsilon}{2(N-1)}$.
Consequently, inequality \eqref{eq:stability1} in the proof of Theorem \ref{thm:throughput_optimal} is established.
\end{proof}

\section{}
\label{app:lm4}
\begin{proof}[Proof of Lemma \ref{lm:successfulq}]
The lemma holds if there is no unsuccessful transmission from time $t$ to $t+q+N-1$.
Otherwise, suppose the first unsuccessful transmission is from user $n$ at time $t_1, t\leq t_1\leq t+q+N-1$.
Since $v(B^{\g}_{t_1})=n$ and no packet is transmitted at time $t_1$, every user will update the upper bound $B^{n,\g}_{t_1+1} = 1$ for user $n$.
From the evolution of the upper bounds we have
\begin{align}
B^{n,\g}_{\tau} = \tau-t_1
\label{eq:Bntau}
\end{align}
for any time $\tau$ if user $n$ is not selected again by $\g$ before time $\tau$.

There are two possibilities: (1) user $n$ is not selected by $\g$ again before time $t+q+N$; (2) user $n$ is selected by $\g$ again at time $t_2$ where $t_1+1\leq t_2\leq t+q+N-1$.

First consider the case when user $n$ is not selected by $\g$ again before time $t+q+N$.
Then \eqref{eq:Bntau} holds for any time $\tau$ from $t_1+1$ to $t+q+N-1$. 
From the specification of $\g$, if any other user $m$ has an unsuccessful transmission at time $t'$, $t_1+1\leq t' \leq t+q+N-1$,
for any subsequent time $\tau \geq t'+1 $ we will have
\begin{align}
B^{m,\g}_{\tau} \leq \tau-t' < \tau-t_1 = B^{n,\g}_{\tau}.
\end{align}
Therefore, user $m$ will not be selected by $\g$ again from time $t'+1$ to $t+q+N-1$. Consequently, any user $m \neq n$ can have at most one unsuccessful transmission from time $t$ to $t+q+N-1$.
Since any of the $N$ users can have at most one unsuccessful transmission from time $t$ to $t+q+N-1$, the number of successful transmissions during this time period is at least $(t+q+N -1)-t+1- N = q$.

Next consider the case when user $n$ is selected by $\g$ again at time $t_2$ where $t_1+1\leq t_2\leq t+q+N-1$.
From the specification of $\g$, we must have $B^{n,\g}_{t_2} = \max_{m} B^{m,\g}_{t_2}$ for user $n$ to transmit at time $t_2$.
Therefore, letting $\tau = t_2$ in \eqref{eq:Bntau} we get
\begin{align}
t_2-t_1 = B^{n,\g}_{t_2} \geq B^{m,\g}_{t_2}
 \label{eq:Bmt2up}
\end{align}
for all $m \neq n$.
Let $S^m, m\neq n$ be the number of successful transmissions for user $m$ between time $t_1$ and $t_2$.
We prove in the following that $S^m \geq Q^{m,\g}_{t_1}$.
\\
If user $m$ has an unsuccessful transmission between $t_1$ and $t_2$, then the queue at user $m$ is empty at the time of the unsuccessful transmission. Therefore,
$S^m \geq Q^{m,\g}_{t_1}$ because all the $Q^{m,\g}_{t_1}$ packets queued at time $t_1$ are successfully transmitted by user $m$ between time $t_1$ and $t_2$. 
\\
If user $m$ transmits successfully in every time slot selected by $\g$, from \eqref{eq:Bmt2up} and the evolution of the upper bounds we obtain
\begin{align}
t_2-t_1 \geq B^{m,\g}_{t_2} = B^{m,\g}_{t_1}+t_2-t_1-S^m.
\label{eq:Smlb2pre}
\end{align}
Since $B^{m,\g}_{t_1} \geq Q^{m,\g}_{t_1}$, \eqref{eq:Smlb2pre} implies
\begin{align}
S^m \geq B^{m,\g}_{t_1} \geq Q^{m,\g}_{t_1}.
\label{eq:Smlb2}
\end{align}
Consequently, for every user $m\neq n$
\begin{align}
S^m \geq Q^{m,\g}_{t_1}.
\end{align}
Note that the total number of successful transmissions between $t_1$ and $t_2$ is $\sum_{m\neq n}S^m $; therefore,
\begin{align}
\sum_{\tau=t_1+1}^{t_2-1} \bar{U}_{\tau} = \sum_{m\neq n}S^m \geq \sum_{m\neq n}Q^{m,\g}_{t_1} = Q^{tot,\g}_{t_1}
\label{eq:Qtott1}
\end{align}
where the last equation in \eqref{eq:Qtott1} holds because $Q^{n,\g}_{t_1}=0$.
\\
From the dynamics of queues we get
\begin{align}
Q^{tot,\g}_{t_1} = & Q^{tot,\g}_{t}+ \sum_{\tau=t}^{t_1-1} \left( \sum_{n=1}^N A^n_{\tau}-\bar{U}_{\tau}\right) \nonumber\\
			     \geq & q- \sum_{\tau=t}^{t_1-1}\bar{U}_{\tau} 
\label{eq:Qtotandt}
\end{align}
Combining \eqref{eq:Qtott1} and \eqref{eq:Qtotandt}, the total number of successful transmissions from time $t$ to $t+q+N-1$ in the second case is at least
\begin{align}
\sum_{\tau=t}^{t+q+N-1} \bar{U}_{\tau}
\geq & \sum_{\tau=t}^{t_1-1} \bar{U}_{\tau}+\sum_{\tau=t_1+1}^{t_2-1} \bar{U}_{\tau} \nonumber\\
\geq & q - Q^{tot,\g}_{t_1} + Q^{tot,\g}_{t_1}  = q.
\end{align}

\end{proof}

\section{}
\label{app:thm:delay}
\begin{proof}[Detailed derivation of \eqref{eq:QTkrec} and \eqref{eq:qt}  in the proof of Theorem \ref{thm:delay}]
\\
\textit{Detailed derivation of \eqref{eq:QTkrec}:}
\\
From \eqref{eq:Tk} and the dynamics of queues we obtain
\begin{align}
Q^{tot,\g}_{T_k} = & Q^{tot,\g}_{T_{k-1}} + \sum_{t=T_{k-1}}^{T_k-1}\left(\sum_{n=1}^N A^n_t - \bar{U}^\g_t\right) \nonumber\\
 = & Q^{tot,\g}_{T_{k-1}} - \sum_{t=T_{k-1}}^{T_k-1} \bar{U}^\g_t+ \sum_{t=T_{k-1}}^{T_k-1}\sum_{n=1}^N A^n_t \nonumber\\
 = &  \sum_{t=T_{k-1}}^{T_k-1}\sum_{n=1}^N A^n_t 
 \label{eq:QTkeqA}\\
 \leq &  \sum_{t=T_{k-1}}^{T_{k-1}+Q^{tot,\g}_{T_{k-1}}+N-1}\sum_{n=1}^N A^n_t .
 \label{eq:QTk}
\end{align}
Equation \eqref{eq:QTkeqA} follows from the definition of $T_k$. 
Inequality \eqref{eq:QTk} is true because of \eqref{eq:Tk} and the fact that $A^n_t$ are all positive.
Note that $A^n_t$ is independent of $T_{k-1}$ and $Q^{tot,\g}_{T_{k-1}}$ for $t\geq T_{k-1}$. 
Therefore, taking the expectation on both sides of \eqref{eq:QTk} we get
\begin{align}
 &\mathbf{E}\left[Q^{tot,\g}_{T_k}\right] \nonumber\\
\leq & \mathbf{E}\left[\sum_{t=T_{k-1}}^{T_{k-1}+Q^{tot,\g}_{T_{k-1}}+N-1}\sum_{n=1}^N A^n_t \right] \nonumber\\
= & \mathbf{E}\left[
\mathbf{E}\left[\left.\sum_{t=T_{k-1}}^{T_{k-1}+Q^{tot,\g}_{T_{k-1}}+N-1}\sum_{n=1}^N A^n_t
\right|T_{k-1},Q^{tot,\g}_{T_{k-1}}\right] 
\right] \nonumber\\
= & \mathbf{E}\left[
\sum_{t=T_{k-1}}^{T_{k-1}+Q^{tot,\g}_{T_{k-1}}+N-1}\sum_{n=1}^N 
\mathbf{E}\left[A^n_t | T_{k-1},Q^{tot,\g}_{T_{k-1}}\right] 
\right] \nonumber\\
= & \mathbf{E}\left[
\sum_{t=T_{k-1}}^{T_{k-1}+Q^{tot,\g}_{T_{k-1}}+N-1}\sum_{n=1}^N \lambda^n\right] \nonumber\\
= & \lambda^{tot}\left(\mathbf{E}\left[Q^{tot,\g}_{T_{k-1}}\right]+N\right).
\end{align}

\textit{Detailed derivation of \eqref{eq:qt}:}
\\
For any time $t=0,1,2,\dots$, suppose $ T_{k-1} < t \leq T_{k} $ ($T_0:= 0$). Using \eqref{eq:QTkrec} and the dynamics of queues we get
\begin{align}
&\mathbf{E}\left[Q^{tot,\g}_{t}\right] \nonumber\\
= & 
\mathbf{E}\left[Q^{tot,\g}_{T_{k-1}} + \sum_{\tau=T_{k-1}}^{t-1}\left(\sum_{n=1}^N A^n_{\tau} - \bar{U}^\g_{\tau}\right)\right] \nonumber\\
\leq & 
\mathbf{E}\left[Q^{tot,\g}_{T_{k-1}} + \sum_{\tau=T_{k-1}}^{t-1}\left(\sum_{n=1}^N A^n_{\tau} \right)\right] \nonumber\\
\leq & 
\mathbf{E}\left[Q^{tot,\g}_{T_{k-1}} + \sum_{\tau=T_{k-1}}^{T_{k}-1}\left(\sum_{n=1}^N A^n_{\tau} \right)\right] \nonumber\\
\stackrel{(a)}{=} & 
\mathbf{E}\left[Q^{tot,\g}_{T_{k-1}} + Q^{tot,\g}_{T_{k}}\right] \nonumber\\
\leq & 2 \frac{\lambda^{tot} N}{1-\lambda^{tot}};
\label{eq:qtapp}
\end{align}
equation (a) in \eqref{eq:qtapp} holds because of \eqref{eq:QTkeqA} and the last inequality in \eqref{eq:qtapp} follows from \eqref{eq:QTkbound}.

\end{proof}

%
%
%
%
%

\section*{Acknowledgment}
This work was partially supported by National Science Foundation (NSF) Grant CCF-1111061 and NASA grant NNX12AO54G.
The authors thank Vijay Subramanian for helpful discussions.

\ifCLASSOPTIONcaptionsoff
  \newpage
\fi



%
%

\bibliographystyle{ieeetr}
\bibliography{MACref}
%

\begin{IEEEbiographynophoto}{Yi Ouyang}(S'13)
received the B.S. degree in Electrical Engineering from the National Taiwan University, Taipei, Taiwan in 2009.
He is currently a Ph.D. student in Electrical Engineering and Computer Science at the University of Michigan, Ann Arbor, MI, USA.
His research interests include stochastic scheduling, decentralized stochastic control and stochastic dynamic games with asymmetric information.

\end{IEEEbiographynophoto}

\begin{IEEEbiographynophoto}{Demosthenis Teneketzis}(M'87--SM'97--F'00)
received the diploma in electrical engineering from the University of Patras, Patras, Greece, and the M.S.,
E.E., and Ph.D. degrees, all in electrical engineering, from the Massachusetts Institute of Technology,
Cambridge, MA, USA, in 1974, 1976, 1977, and 1979, respectively.

He is currently Professor of Electrical Engineering and Computer Science at the University of Michigan,
Ann Arbor, MI, USA. In winter and spring 1992, he was a Visiting Professor at the Swiss Federal Institute
of Technology (ETH), Zurich, Switzerland. Prior to joining the University of Michigan, he worked for Systems Control, Inc., Palo Alto, CA,USA, and Alphatech, Inc., Burlington, MA, USA.
His research interests are in stochastic control, decentralized systems, queueing and communication networks, stochastic scheduling and resource allocation problems, mathematical economics, and discrete-event systems.
\end{IEEEbiographynophoto}






\end{document}